\documentclass[10pt,psamsfonts]{amsart}
\usepackage{eucal}
\usepackage{amsmath,amssymb,amsfonts,amsthm,mathrsfs}
\usepackage{graphicx, array}
\usepackage{enumerate}
\usepackage{epsfig}
\usepackage{float}
\usepackage{color}
\usepackage{xcolor}
\colorlet{BLUE}{blue}
\usepackage[foot]{amsaddr}
\usepackage{enumitem}
\usepackage{multicol}
\usepackage[colorlinks=true, urlcolor=blue, linkcolor=red]{hyperref}

\usepackage[mathlines]{lineno}

\newtheorem{theorem}{Theorem}[section]
\newtheorem{lemma}{Lemma}[section]
\theoremstyle{corollary}

\theoremstyle{definition}
\newtheorem{definition}{Definition}[section]

\numberwithin{equation}{section}

\def\longdelete#1{}


\begin{document}

\title[Savanna Dynamics]{Savanna dynamics with grazing, browsing, and migration effects}

\author[Chen, C-C]{Chiun-Chuan Chen$^{1,3}$}
\address{$^1$ Department of Mathematics, National Taiwan University, Taiwan}
\email[Chen, C-C]{$^{1,3}$chchchen@math.ntu.edu.tw}

\author[Hsiao, T-Y]{Ting-Yang Hsiao$^2$}
\address{$^2$ (corresponding author) University of Illinois, Urbana-Champaign, Department of Mathematics, USA}
\email[Hsiao, T-Y]{$^2$tyhsiao2@illinois.edu}

\author[Wang, S-C]{Shun-Chieh Wang$^3$}
\address{$^3$ National Center for Theoretical Science, Taipei, Taiwan}
\email[Wang, S-C]{$^3$rjaywang1130@ncts.ntu.edu.tw}


\date{\today}

\maketitle

\begin{abstract}
This article explores the dynamics of savanna ecosystems with grazing, browsing, and migration effects. Covering over one-eighth of the Earth’s land area and supporting about one-fifth of the global population, the savanna is an ecological system whose importance has only recently garnered significant attention from biologists. The interactions between organisms in this ecosystem are highly complex, and fundamental mathematical issues remain unresolved. We rigorously analyze traveling waves in savanna systems and focus on whether trees, grass, grazers, and browsers coexist. We demonstrate the existence of various traveling waves, including waves transitioning from extinction to co-existence and waves from a grass-vegetation state (where only grass and grazers exist) to co-existence. Due to the biodiversity of species in grassland ecosystems, it is not appropriate to consider overly simplified models of competition between grasses and trees. From both a biological and mathematical perspective, factors such as animal grazing, browsing, and migration (which facilitates seed dispersal) play a crucial role in promoting ecological stability and coexistence. Additionally, we estimate the nonzero minimum value of the total plant biomass within the savanna dynamic system to better understand the persistence and stability of sustainable development within the ecosystem.
\end{abstract}


\section{Introduction}
In a savanna (mixed woodland-grassland), animals play a crucial role in maintaining the balance of ecosystems by grazing, browsing, and migration (diffusion). First, herbaceous plants and woody plants compete for water in the upper soil layers, and grazing reduces this competition, potentially promoting tree growth. Second, the removal of herbaceous plants as fuel reduces the intensity and frequency of fires, helping to spread the species of woody plants. Third, animals also have direct impacts on herbaceous plants and woody plants by grazing and browsing edible plants. Evidence suggests that in savannas, woody plants with poor palatability increase in grazing areas. Grazing, browsing, and migration also play a role in spreading the seeds of herbaceous plants and woody plants. For instance, cattle and horses contribute to the seed dispersal of certain weed species, such as Acacia nilotica and species in the Stylosanthes genus. The changes in species composition caused by grazing and browsing in savannas alter ecosystem functions. The research on savanna dynamic systems can be referenced in the following materials \cite{gordon2008ecology,gordon2019ecology}.

From a biological perspective, in natural community, trees and herbaceous plants may not coexist. Many biological studies have observed this phenomenon. Ecosystems, where trees and herbaceous plants coexist, can suddenly shift to tree-dominated ecosystems (known as shrub encroachment), which has become a serious problem \cite{brown1998spatial,van2000shrub,groen2017spatially}. On the other hand, an overabundance of herbaceous plants increases the likelihood of wildfires, which leads to the death of many trees \cite{beckage2011grass}. Both of these effects contribute to the inability of trees and herbaceous plants to coexist.

The co-existence of trees and herbaceous plants can be explained by niche separation mechanisms, where plants differentiate along environmental axes such as light, soil moisture, and root depth \cite{sydes1984comparative}. Also, the two-layer hypothesis \cite{walter1971ecology} explains the long-term co-existence of trees and grasses in grasslands, suggesting that differences in root depth allow trees to access deeper soil layers while grasses dominate the upper layers. This results in weak competition between the two, with grasses having an advantage in the upper soil. Under this biological effect, one may assume $\kappa_1>0$ and $\kappa_2>0$ in the model \eqref{o W} below are small.

In human-managed systems, a reasonable assumption is that grazing and browsing pressure is constant. 
Co-existence states based on this assumption can be observed in previous studies \cite{dublin1990elephants, van2003effects, d2006probabilistic,groen2017spatially}. Also, ignoring the effects of natural fires and the water resources brought by seasonal rainfall (for instance, see \cite{klimasara2023model}) is reasonable in human-managed systems since the grazing, browsing, and migration of animals decrease the possibility of wildfire and the human-managed environment can provide a stable water supply. In this paper, we aim to better understand the co-existence state. Taking into account the effects of animal and plant migration, as well as the impacts of grazing and browsing, we adopt the models considered in \cite{klimasara2023model} and \cite{beckage2011grass}, and propose the following system with migration effects:
\begin{equation} \label{o W}
    \begin{aligned} 
    W_t&=d_W W_{xx}+r_W W(1-\frac{W}{K_W}-\frac{\kappa_1 H}{K_H})-c_W BW,\\
    H_t&=d_H H_{xx}+r_H H(1-\frac{H}{K_H}-\frac{\kappa_2 W}{K_W})-c_H GH,\\
    G_t&=d_G G_{xx}+e_H GH-a_G G^2,\\ 
    B_t&=d_B B_{xx}+e_W BW-a_B B^2,
\end{aligned}
\end{equation}
where $W, H, G,$ and $B$ are the amount of woody, grass, grazers, and browsers biomass respectively. Also, $r_W$ and $r_H$ represent the growth rates, $K_W$ and $K_H$ represent the carrying capacities, and $c_W$ and $c_H$ denote the consumption
coefficients of trees and grasses respectively. Additionally, $e_W$ and $e_H$ are the conversion efficiency rates of woody and grass biomass by browsers and grazers, while $a_B$ and $a_G$ stand for the intraspecific competition coefficients of browsers and grazers. All parameters above are assumed to be positive.

Let us specifically comment on migration effects and relative interspecific competition coefficients $\kappa_1$ and $\kappa_2$. First, $d_G$ and $d_B$, caused by migration, are the diffusion rates of grazers and browsers respectively, and $d_W$ and $d_H$ are (seed) dispersal rates of trees and herbaceous plants respectively. Second, the positive parameters $\kappa_1$ and $\kappa_2$ are referred to as the relative interspecific competition coefficients. 
If we neglect migration effects and assume $\kappa_1=0$ and $\kappa_2=1$, then \eqref{o W} is reduced to the model in \cite{klimasara2023model} (see also \cite{beckage2011grass}). More general cases, $\kappa_1>0$ and $\kappa_2>0$, are explored in this article (Theorem \ref{main 2.1}). 

We are interested in what would happen if grazers and browsers are added into a purely competitive environment among plants and whether trees, grass, grazers, and browsers can coexist. Therefore, we focus on the co-existence steady state of four species (which is unique) in the proposed savanna system \eqref{o W} and study related traveling wave solutions, which may provide important information about the transition dynamics between the co-existence state and other steady states. We demonstrate the existence of two types of traveling waves, the waves transitioning
from the extinction state to the co-existence state and the waves from a grass-vegetation state (where only grass and grazers exist) to the co-existence state. Our results indicate that under suitable conditions, the co-existence state of four species may be dominant or more stable compared to other steady states in the sense of wave dynamics. To obtain our results, we employ the upper-lower solution method widely used in related studies. However, due to the absence of linear reaction terms in our herbivore equations for $G$ and $B$, the usual form of exponential functions for the lower solutions does not work in our case. To overcome this difficulty, we introduce Gaussian-like functions at one end of the real line domain and successfully find suitable lower solutions for our problem. 

In ecology, how to estimate the total biomass of the organisms in a system is an important issue. For the waves connecting the grass-vegetation state and the co-existence state of four species, since the total mass of the plants (grass and wood) are positive at infinities, it is interesting to ask whether the total plant mass has a lower bound which only depends on the parameters of the equations. Applying an N-barrier maximum (see \cite{chen2016maximum} for instance), we are able to obtain such a lower bound for plants via a new construction of the “barrier”. This result is stated in Theorem \ref{main thm 3}. 

The paper is divided as follows. We first propose a Savanna model \eqref{o W} with the effects of grazing, browsing, and migration. In Section \ref{Sec 2}, we give a basic setup for this reaction-diffusion system and state the results. In Sections \ref{sec 3} and \ref{section 4}, we prove the existence of traveling waves, which relies primarily on the construction of upper and lower solutions and the application of the Schauder fixed-point theorem. As mentioned, due to the absence of linear reaction terms in the herbivore dynamics, the main result cannot be fully appreciated until the Gaussian-like functions are introduced. Additionally, we use the shrinking box argument to characterize the asymptotic behavior of the traveling wave solutions. This shows the existence of smooth positive traveling wave solutions connecting extinction to co-existence states (Theorem \ref{main 2.1}). In Section \ref{sec 5}, we neglect the effect of wildfire, $\kappa_1=0$, (see \cite{klimasara2023model} for more biological reasons). Such an assumption allows us to first solve for $(W,B)$, and then solve for $(H,G)$. We show the existence of smooth positive traveling wave solutions connecting grass-vegetation to co-existence states (Theorem \ref{main 2.2}). Also, we estimate the total biomass of plants in Theorem \ref{main thm 3}.

\section{Traveling waves of savanna dynamics} \label{Sec 2}
Let us now introduce the dimensionless
variables
\begin{align*}
    \hat{W}=\frac{W}{K_W}, ~\hat{H}=\frac{H}{K_H},~\hat{G}=\frac{a_G G}{e_H K_H},~\hat{B}=\frac{a_B B}{e_W K_W},
\end{align*}
and
\begin{align*}
    \hat{c}_W=\frac{c_W e_W K_W}{a_B},~\hat{c}_H=\frac{c_H e_H K_H}{a_G},~\hat{e}_H=e_H K_H,~\hat{e}_W=e_W K_W.
\end{align*}
For notational simplicity, let us drop the $\hat{\cdot}$ and rewrite \eqref{o W} into the following dimensionless system:
\begin{equation} \label{W}
    \begin{aligned}
    W_t&=d_W W_{xx}+r_W W(1-W-\kappa_1 H)-c_W BW,\\ 
    H_t&=d_H H_{xx}+r_H H(1-H-\kappa_2 W)-c_H GH,\\ 
    G_t&=d_G G_{xx}+e_H G(H-G),\\ 
    B_t&=d_B B_{xx}+e_W B(W-B).
\end{aligned}
\end{equation}

In the spatially homogeneous case, a straightforward computation reveals that the only steady state of co-existence (strictly positive constant solutions) in the system \eqref{W} is
\begin{equation} \label{co existence}
    \begin{aligned}
        &W^*=\frac{r_W(r_H+c_H)-r_W r_H \kappa_1}{(r_W+c_W)(r_H+c_H)-r_W r_H  \kappa_1\kappa_2},\\
        &H^*=\frac{r_H(r_W+c_W)-r_W r_H \kappa_2}{(r_W+c_W)(r_H+c_H)-r_W r_H  \kappa_1\kappa_2},\\
        &G^*=\frac{r_H(r_W+c_W)-r_W r_H \kappa_2}{(r_W+c_W)(r_H+c_H)-r_W r_H  \kappa_1\kappa_2},\\
        &B^*=\frac{r_W(r_H+c_H)-r_W r_H \kappa_1}{(r_W+c_W)(r_H+c_H)-r_W r_H  \kappa_1\kappa_2}.
    \end{aligned}
\end{equation}

In this article, we assume that the competition between $W$ and $H$ is weak, i.e., $\kappa_1<1$ and $\kappa_2<1$, which ensures that \eqref{co existence} is well-defined with a positive solution.

Another biologically significant state is the grass vegetation where trees and browsers die out but herbaceous plants and grazers survive:
\begin{align*}
    (W_*,H_*,G_*,B_*)=(0,\frac{r_H}{r_H+c_H},\frac{r_H}{r_H+c_H},0).
\end{align*}
One interesting phenomenon is the emergence of traveling wave solutions in the system \eqref{W}, which represent waves of persistent configuration that appear to be independent of time when observed from a moving reference frame. In particular, they demonstrate the ansatz:
\begin{align} \notag
W(x,t)&=W^s(x+st),~H(x,t)=H^s(x+st),
\end{align}
and
\begin{align} \notag
G(x,t)&=G^s(x+st),~~~~~~~B(x,t)=B^s(x+st),
\end{align}
for some traveling wave profile $(W^s,H^s,G^s,B^s)$ and wave speed $s\in\mathbb{R}$.

In this article, we consider two types of traveling wave solutions. First, we are interested in
the existence of the traveling wave solutions of system \eqref{W}, connecting the extinction state $(0,0,0,0)$ and co-existence state $(W^*,H^*,G^*,B^*)$. To this end, denoting $\xi=x+st$ in $\mathbb{R}$, we study the solutions of 
\begin{equation} \label{s W}
    \begin{aligned} 
    d_W W^s_{\xi\xi}&-sW^s_{\xi}+r_W W^s(1-W^s-\kappa_1 H^s)-c_W B^sW^s=0,\\ 
    d_H H^s_{\xi\xi}&-sH^s_{\xi}+r_H H^s(1-H^s-\kappa_2 W^s)-c_H G^sH^s=0,\\ 
    d_G G^s_{\xi\xi}&-sG^s_{\xi}+e_H G^s(H^s-G^s)=0,\\ 
    d_B B^s_{\xi\xi}&-sB^s_{\xi}+e_W B^s(W^s-B^s)=0,
\end{aligned}
\end{equation}
with asymptotic behaviors
\begin{align} \label{limit left}
    \lim\limits_{\xi\rightarrow -\infty}(W^s,H^s,G^s,B^s)=(0,0,0,0),
\end{align}
and
\begin{align}  \label{limit right}
    \lim\limits_{\xi\rightarrow +\infty}(W^s,H^s,G^s,B^s)=(W^*,H^*,G^*,B^*).
\end{align}

Second, we show existence of the traveling wave solutions of system \eqref{W}, connecting the state $(0,H_*,G_*,0)$ and co-existence state $(W^*,H^*,G^*,B^*)$. In other words, the solutions satisfy \eqref{s W} with the asymptotic behaviors \eqref{limit right} and 
\begin{align} \label{coe to coe}
    \lim_{\xi\rightarrow-\infty}(W^s,H^s,G^s,B^s)=(0,H_*,G_*,0).
\end{align}

Let 
\begin{align*} 
    s_m=\max\left\{2\sqrt{d_W r_W}, 2\sqrt{d_H r_H}\right\},
\end{align*}

\begin{align*} 
    s_{M_1}=\min\left\{d_B \sqrt{\frac{r_W}{d_B-d_W}}, 2\sqrt{d_B e_W\left(1-\kappa_1-\frac{c_W}{r_W}\right)}\right\},
\end{align*}
and
\begin{align*}
   s_{M_2}=\min\left\{d_G \sqrt{\frac{r_H}{d_G-d_H}},2\sqrt{d_G e_H\left(1-\kappa_2-\frac{c_H}{r_H}\right)}\right\}.
\end{align*}
Our main theorems are outlined as follows.

\begin{theorem} \label{main 2.1} Assume $d_B>2 d_W>0$, $d_G>2 d_H>0$, $1-\kappa_1-\frac{c_W}{r_W}>0$, $1-\kappa_2-\frac{c_H}{r_H}>0$ and $s_m<\min\{s_{M_1},s_{M_2}\}$. Let $s$ be a constant satisfying
\begin{align*} 
    s_m<s<\min\{s_{M_1},s_{M_2}\}.
\end{align*}
Then \eqref{W} admits a positive traveling wave solution $(W^s,H^s,G^s,B^s)\in C^{\infty}(\mathbb{R},\mathbb{R}^4)$ satisfying \eqref{limit left} and \eqref{limit right}.
\end{theorem}

\begin{theorem} \label{main 2.2} Assume 
$\kappa_1=0$, $d_B>2d_W>0$, $1-\frac{c_W}{r_W}>0$, $1-\kappa_2-\frac{c_H}{r_H}>0$ and $s_m<s_{M_1}$. Let $s$ be a constant satisfying
\begin{align*} 
    s_m<s<s_{M_1}.
\end{align*}
Then \eqref{W} admits a positive traveling wave solution $(W^s,H^s,G^s,B^s)\in C^{\infty}(\mathbb{R},\mathbb{R}^4)$ satisfying \eqref{coe to coe} and \eqref{limit right}.
\end{theorem}

We pause to remake that biologically realistic is that the dispersal of plants (seed dispersal) can occur through wind, water, or even by animals. Therefore, the dispersal coefficients, $d_W$ and $d_H$, may not be large but are usually non-zero. Also, since grazers and browsers ``move" relatively faster than plants, it is reasonable to propose $d_B>2d_W>0 ~\mbox{or}~d_G>2d_H>0$, which implies $d_B \sqrt{r_W/(d_B-d_W)}>2\sqrt{d_W r_W}$ or $d_G \sqrt{r_H/(d_G-d_H)}>2\sqrt{d_H r_H}$.

Our next goal is to present a quantitative analysis of the solutions in Theorem \ref{main 2.2}. This estimation helps us to understand the total biomass in the process of establishing a new steady state.

In biology (see, for instance, \cite{klimasara2023model, gordon2019ecology}), system \eqref{s W} serves as a model for savanna vegetation ecology (the food sources for livestock raised by humans) where the total amount of resources represents a significant biomass. In Theorem \ref{main 2.2}, we have already described the transition process of establishing a new steady state: from a purely herbaceous ecosystem to an ecosystem with coexisting herbs and woody plants. We now proceed to characterize a non-trivial lower bound for the total resources within this system. Let 
\begin{align} \label{q p}
    q(\xi)=c_1 W(\xi)+c_2 H(\xi)~~\mbox{and}~~p(\xi)=\frac{c_1}{d_W} W(\xi)+\frac{c_2}{d_H} H(\xi).
\end{align}
In this context, $c_1$ and $c_2$ represent the arithmetic weights of trees and forage (grass), respectively. For example, $c_1$ and $c_2$ could represent the price per unit weight of trees and forage. If $c_1 = c_2 = 1$, we can also interpret $q$ as the total weight of all plants.

The quantity $q$ can be estimated as follows.
\begin{theorem} \label{main thm 3}
     Given $c_1,c_2>0$, suppose that $s>0$, $(W,H,G,B)\in C^{2}(\mathbb{R},\mathbb{R}^4)$ is positive and satisfies the traveling wave system \eqref{s W} with asymptotic behaviors \eqref{coe to coe} and \eqref{limit right}. Then we have, for $\xi\in\mathbb{R}$, 
     \begin{align} \label{minimum of p}
         q(\xi)\geq d\min\left\{c_1 \tilde{W},c_2\tilde{H}\right\},
     \end{align}
where
\begin{align} \label{tilde W H}
    d=\min\left\{\frac{d_W}{d_H},\frac{d_H}{d_W}\right\},~\tilde{H}=\min
    \left\{1-\frac{c_H}{r_H},\frac{1}{\kappa_1}\left(1-\frac{c_W}{r_W}\right)\right\},
\end{align}  
and
\begin{align}
    \tilde{W}=\min\left\{1-\frac{c_W}{r_W},\frac{1}{\kappa_2}\left(1-\frac{c_H}{r_H}\right)\right\}.
\end{align}
\end{theorem}

\section{Preliminaries of traveling wave solutions} \label{sec 3}
Throughout this paper, we drop the superscript $s$ for notational simplicity. We present certain properties of system \eqref{s W} with asymptotic behaviors \eqref{limit left} and \eqref{limit right} or \eqref{coe to coe} and \eqref{limit right}. 
\begin{lemma}
    If the quadruple $(W(\xi),H(\xi),G(\xi),B(\xi))$ is a strictly positive solution of \eqref{s W}-\eqref{limit right}, then $0<W(\xi),H(\xi),G(\xi),B(\xi)<1$, for $\xi\in\mathbb{R}$.
\end{lemma}
\begin{proof}
    The lower bound is based on the assumption of positivity. Let us demonstrate that $W(\xi)<1$, for $\xi\in\mathbb{R
    }$. Suppose otherwise; since the right asymptotic behavior of $W$ is strictly less than $1$, there exists a point $\xi_0\in\mathbb{R}$ such that $W(\xi_0)\geq 1$, being a global maximum of $W(\xi)$. Then we have $W_\xi(\xi_0)=0$, $W_{\xi\xi}(\xi_0)\leq 0$, and  
    \begin{align*}
        r_W W(\xi_0)(1-W(\xi_0)-\kappa_1 H(\xi_0))-c_W B(\xi_0)W(\xi_0)<0
    \end{align*}
    since $H(\xi)>0$ and $B(\xi)>0$ for $\xi\in\mathbb{R}$, which lead to a contradiction in the first equation of \eqref{s W}. Following a similar argument, we can prove that $H(\xi),G(\xi)$ and $B(\xi)$ are strictly less than $1$, and the lemma follows.
\end{proof}
We pause to note that this property can also be achieved for traveling wave solutions connecting the state $(0,H_*,G_*,0)$ and co-existence state $(W^*,H^*,G^*,B^*)$ by using a similar argument.

\subsection{Upper and lower solutions}

Following \cite{ma2001traveling} and others, we define the upper and lower solutions as follows.

\begin{definition} \label{UL for WB}
    The functions $(\overline{W},\overline{H},\overline{G},\overline{B})$ and $(\underline{W},\underline{H}, \underline{G},\underline{B})$ are upper and lower solutions of system \eqref{s W} if there exists a $U=\{\xi_1,...,\xi_n\}$ such that
\begin{enumerate}[label=(H\arabic*)]
    \item the functions $\overline{W},\underline{W},\overline{H},\underline{H},\overline{G},\underline{G},\overline{B}$, and $\underline{B}$ are $C^2$ functions in $\mathbb{R}\setminus U$;
    \item the first and second derivatives of $\overline{W},\underline{W},\overline{H},\underline{H},\overline{G},\underline{G},\overline{B}$, and $\underline{B}$ are bounded in $\mathbb{R} \setminus U$;
    \item for $i=1,...,n$, $\overline{W}_{\xi}(\xi^+_i)\leq \overline{W}_{\xi}(\xi^-_i)$, $\overline{H}_{\xi}(\xi^+_i)\leq \overline{H}_{\xi}(\xi^-_i)$, $\overline{G}_{\xi}(\xi^+_i)\leq \overline{G}_{\xi}(\xi^-_i)$, and $\overline{B}_{\xi}(\xi^+_i)\leq \overline{B}_{\xi}(\xi^-_i)$;
    \item for $i=1,...,n$, $\underline{W}_{\xi}(\xi^-_i)\leq \underline{W}_{\xi}(\xi^+_i)$, $\underline{H}_{\xi}(\xi^-_i)\leq \underline{H}_{\xi}(\xi^+_i)$, $\underline{G}_{\xi}(\xi^-_i)\leq \underline{G}_{\xi}(\xi^+_i)$, and $\underline{B}_{\xi}(\xi^-_i)\leq \underline{B}_{\xi}(\xi^+_i)$;
    \item the functions $(\overline{W},\overline{H},\overline{G},\overline{B})$ and $(\underline{W},\underline{H}, \underline{G},\underline{B})$ satisfy
\begin{align} \label{upper W}
    d_W \overline{W}_{\xi\xi}&-s\overline{W}_{\xi}+r_W \overline{W}(1-\overline{W}-\kappa_1 \underline{H})-c_W \underline{B}\overline{W} \leq 0,\\ \label{lower W}
    d_W \underline{W}_{\xi\xi}&-s\underline{W}_{\xi}+r_W \underline{W}(1-\underline{W}-\kappa_1 \overline{H})-c_W \overline{B}\underline{W} \geq 0,\\ \label{upper H}
    d_H \overline{H}_{\xi\xi}&-s\overline{H}_{\xi}+r_H \overline{H}(1-\overline{H}-\kappa_2 \underline{W})-c_H \underline{G}\overline{H} \leq 0,\\ \label{lower H}
    d_H \underline{H}_{\xi\xi}&-s\underline{H}_{\xi}+r_H \underline{H}(1-\underline{H}-\kappa_2 \overline{W})-c_H \overline{G}\underline{H} \geq 0,\\ \label{upper G}
    d_G \overline{G}_{\xi\xi}&-s\overline{G}_{\xi}+e_H \overline{G}(\overline{H}-\overline{G}) \leq 0,\\ \label{lower G}
    d_G \underline{G}_{\xi\xi}&-s\underline{G}_{\xi}+e_H \underline{G}(\underline{H}-\underline{G}) \geq 0,\\
    \label{upper B}
    d_B \overline{B}_{\xi\xi}&-s\overline{B}_{\xi}+e_W \overline{B}(\overline{W}-\overline{B}) \leq 0,\\ \label{lower B}
    d_B \underline{B}_{\xi\xi}&-s\underline{B}_{\xi}+e_W \underline{B}(\underline{W}-\underline{B}) \geq 0,
\end{align}
in the region $\mathbb{R} \setminus U$.
\end{enumerate}
\end{definition}

Our next goal is to construct upper and lower solutions for system \eqref{s W}. To this end, we define constants as follows:
\begin{align} \label{lambda W}
    \lambda_W=\frac{s-\sqrt{s^2-4d_W r_W}}{2d_W},~\eta_W=\frac{s+\sqrt{s^2-4d_W r_W}}{2d_W},
\end{align}
\begin{align} \label{lambda H}
    \lambda_H=\frac{s-\sqrt{s^2-4d_H r_H}}{2d_H},~\eta_H=\frac{s+\sqrt{s^2-4d_H r_H}}{2d_H},
\end{align}
\begin{align} \label{lambda B}
    \lambda_B=\frac{s}{d_B},
\end{align}
and 
\begin{align} \label{lambda G}
    \lambda_G=\frac{s}{d_G}.
\end{align}
We pause to remark that \eqref{lambda W} and \eqref{lambda H} are well defined provided that $s>2\sqrt{d_W r_W}$ and $s>2\sqrt{d_H r_H}$ respectively. We notice that $\lambda_W,\eta_W$ are the roots of the quadratic polynomial $d_W x^2-s x+r_W=0$, $\lambda_H,\eta_H$ are the roots of the quadratic polynomial $d_H x^2-s x+r_H=0$, $\lambda_B$ is a root of the quadratic polynomial $d_B x^2-s x=0$ and $\lambda_G$ is a root of the quadratic polynomial $d_G x^2-s x=0$. 

Considering the assumptions in Theorem \ref{main 2.1}, we choose $\varpi$, $\delta_W$, $\omega$, $\beta$, $b$, $h$, $\delta_H$, $\eta$, $\gamma$, and $g$ satisfy the following conditions:
\begin{align} \label{1-cw/rw}
   &0<\varpi<1-\kappa_1-\frac{c_W}{r_W},\\ \label{delta_w condition}
   &\delta_W\in(1,\min\left\{2,\frac{\eta_W} {\lambda_W},\frac{\lambda_W+\lambda_B}{\lambda_W},\frac{\lambda_W+\lambda_H}{\lambda_W}\right\}),\\ \label{omega condition}
   &\omega>\frac{r_W+r_W\kappa_1+c_W}{-(d_W(\delta_W \lambda_W)^2-s(\delta_W\lambda_W)+r_W)},\\ \label{2.17 condition}
   &-2d_B\beta-\frac{s^2}{4d_B}+e_W(\varpi-b)>0,\\
   \label{1-w-ch/rh}
   &0<h<1-\kappa_2-\frac{c_H}{r_H},\\ \label{delta_h condition}
   &\delta_H\in(1,\min\left\{2,\frac{\eta_H} {\lambda_H},\frac{\lambda_H+\lambda_G}{\lambda_H},\frac{\lambda_H+\lambda_W}{\lambda_H}\right\}),\\ \label{eta condition}
   &\eta>\frac{r_H+r_H\kappa_2+c_H}{-(d_H(\delta_H \lambda_H)^2-s(\delta_H\lambda_H)+r_H)},\\ \label{2.31 condition}
   &-2d_G\gamma-\frac{s^2}{4d_G}+e_H(h-g)>0.
\end{align}

Recall the assumptions in Theorem \ref{main 2.1} and the definitions of $s_{M_1}$ and 
  $s_{M_2}$. We notice that by choosing small enough $\beta$, $b$, $\gamma$, and $g$ and choosing $\varpi$ and $h$ close enough to $1-\kappa_1-\frac{c_W}{r_W}$ and $1-\kappa_2-\frac{c_H}{r_H}$ respectively. One can obtain \eqref{1-cw/rw}, \eqref{2.17 condition}, \eqref{1-w-ch/rh}, and \eqref{2.31 condition}. Also note that the denominators in \eqref{omega condition} and \eqref{eta condition} are positive and it is easy to choose suitable $\delta_W$, 
$\omega$, $\delta_H$, and $\eta$. 

We are ready to introduce the upper functions $(\overline{W},\overline{H},\overline{G},\overline{B})$ and the lower functions $(\underline{W},\underline{H},\underline{G},\underline{B})$ as follows.

\begin{lemma} \label{Lemma UL WB}
    Suppose that the assumptions described in Theorem \ref{main 2.1} hold. For $s_m<s<\min\{s_{M_1},s_{M_2}\}$, if \eqref{1-cw/rw}-\eqref{2.31 condition} hold true, then the functions 
    \begin{equation} \label{overline W}
\overline{W}(\xi)=\left\{ \begin{aligned}
    &1 & \xi\geq 0,\\
    &e^{\lambda_W \xi}& \xi\leq 0,
\end{aligned}
\right.
\end{equation}

\begin{equation} \label{overline H}
\overline{H}(\xi)=\left\{ \begin{aligned}
    &1 & \xi\geq 0,\\
    &e^{\lambda_H \xi}& \xi\leq 0,
\end{aligned}
\right.
\end{equation}

\begin{equation} \label{overline G}
\overline{G}(\xi)=\left\{ \begin{aligned}
    &1 & \xi\geq 0,\\
    &e^{\lambda_G \xi}& \xi\leq 0,
\end{aligned}
\right.
\end{equation}

\begin{equation} \label{overline B}
\overline{B}(\xi)=\left\{ \begin{aligned}
    &1 & \xi\geq 0,\\
    &e^{\lambda_B \xi}& \xi\leq 0,
\end{aligned}
\right.
\end{equation}

\begin{equation} \label{underline W}
\underline{W}(\xi)=\left\{ \begin{aligned}
    &\varpi & \xi\geq \xi_W,\\
    &e^{\lambda_W \xi}-\omega e^{\delta_W\lambda_W \xi}& \xi\leq \xi_W,
\end{aligned}
\right.
\end{equation}

\begin{equation} \label{underline H}
\underline{H}(\xi)=\left\{ \begin{aligned}
    &h & \xi\geq \xi_H,\\
    &e^{\lambda_H \xi}-\eta e^{\delta_H\lambda_H \xi}& \xi\leq \xi_H,
\end{aligned}
\right.
\end{equation}

\begin{equation} \label{underline G}
\underline{G}(\xi)=\left\{ \begin{aligned}
    &g & \xi\geq \xi_G,\\
    &g e^{-\gamma(\xi-\xi_G)^2}& \xi\leq \xi_G,
\end{aligned}
\right.
\end{equation}

\begin{equation} \label{underline B}
\underline{B}(\xi)=\left\{ \begin{aligned}
    &b & \xi\geq \xi_B,\\
    &b e^{-\beta(\xi-\xi_B)^2}& \xi\leq \xi_B,
\end{aligned}
\right.
\end{equation}
satisfy Definition \ref{UL for WB}, where $\xi_W<0<\xi_B$ and $\xi_H<0<\xi_G$ for some $\xi_W,\xi_H$ and some large enough $\xi_B,\xi_G$.
\end{lemma}

\begin{proof}
  Clearly, the domain $U$ is chosen by $\{\xi_H, \xi_W, 0, \xi_B,\xi_G\}$ and the proofs of (H1)-(H4) are rudimentary. First, we claim that inequality \eqref{upper W} holds. Combining the facts that 
\begin{align*}
    d_W \lambda^2_W-s \lambda_W+r_W=0,
\end{align*}
and 
\begin{align*}
    \overline{W}>0,~~\underline{H}>0,\mbox{and}~\underline{B}>0,
\end{align*}
we verify inequality \eqref{upper W}. Recalling that $d_B>2d_W>0$ and $s_m<s<s_{M_1}$, we obtain that $e^{\lambda_B \xi}-e^{\lambda_W \xi}\geq 0$, when $\xi\leq 0$. Then we may follow a similar manner of argument to verify inequality \eqref{upper B}. Next, we show inequality \eqref{lower W} holds. When $\xi\geq \xi_W$, recalling \eqref{1-cw/rw}, we obtain
\begin{align*}
    &d_W \underline{W}_{\xi\xi}-s\underline{W}_{\xi}+r_W \underline{W}(1-\underline{W}-\kappa_1 \overline{H})-c_W \overline{B}\underline{W}\\
    &\geq r_W \varpi(1-\varpi-\kappa_1)-c_W \overline{B} \varpi>0.
\end{align*}
Then we confirm the inequality \eqref{lower W}. For the case when $\xi<\xi_W$, we arrive at
\begin{equation} \label{3.1 729}
\begin{aligned}
     d_W& \underline{W}_{\xi\xi}-s\underline{W}_{\xi}+r_W \underline{W}(1-\underline{W}-\kappa_1\overline{H})-c_W \overline{B}\underline{W}\\
     =&\omega e^{\delta_W \lambda_W \xi}\Big(-d_W(\delta_W \lambda_W)^2+s(\delta_W\lambda_W)-r_W\Big)\\
     &-r_W(e^{\lambda_W \xi}-\omega e^{\delta_W\lambda_W \xi})^2-(r_W\kappa_1 e^{\lambda_H \xi}+c_W e^{\lambda_B\xi})(e^{\lambda_W \xi}-\omega e^{\delta_W\lambda_W \xi})\\
     \geq&\omega e^{\delta_W \lambda_W \xi}\Big(-d_W(\delta_W \lambda_W)^2+s(\delta_W\lambda_W)-r_W\Big)\\
     &-r_W e^{2\lambda_W \xi}-r_W\kappa_1 e^{(\lambda_W+\lambda_H)\xi}-c_W e^{(\lambda_W+\lambda_B) \xi}\\
     =&e^{\delta_W \lambda_W \xi}\Big(-\omega(d_W(\delta_W \lambda_W)^2-s(\delta_W\lambda_W)+r_W)\\
     &-r_W e^{(2-\delta_W)\lambda_W \xi}-r_W\kappa_1 e^{((1-\delta_W)\lambda_W+\lambda_H)\xi} -c_We^{((1-\delta_W)\lambda_W+\lambda_B)\xi}\Big)\\
     \geq&e^{\delta_W \lambda_W \xi}\Big(-\omega(d_W(\delta_W \lambda_W)^2-s(\delta_W\lambda_W)+r_W)-r_W-r_W\kappa_1-c_W\Big)\\
     \geq& 0.
\end{aligned}
\end{equation}
We pause to remark on why \eqref{delta_w condition} and \eqref{omega condition} asserts \eqref{3.1 729}. Here the first inequality uses $\eqref{underline W}$ is positive. The second inequality uses 
\begin{align*}
    \delta_W \leq \min\{2,\frac{\lambda_W+\lambda_B}{\lambda_W},\frac{\lambda_W+\lambda_H}{\lambda_W}\}.
\end{align*}
To obtain the third inequality, we use 
\begin{align*}
    \lambda_W<\delta_W \lambda_W<\eta_W,
\end{align*}
and, hence, 
\begin{align*}
    -\left(d_W(\delta_W\lambda_W)^2-s(\delta_W\lambda_W)+r_W\right)>0.
\end{align*}
then for sufficiently large $\omega$ asserts the third inequality. 

Our next step is to vertify inequality \eqref{lower B}. It is obvious that inequality \eqref{lower B} holds when $\xi\geq \xi_B$. When $\xi_W \leq \xi\leq\xi_B$, a straightforward computation achieves that 
\begin{equation}
    \begin{aligned}
        d_B& \underline{B}_{\xi\xi}-s\underline{B}_{\xi}+e_W \underline{B}(\underline{W}-\underline{B})\\
        =&be^{-\beta(\xi-\xi_B)^2}\left(d_B(4\beta^2(\xi-\xi_B)^2-2\beta)+2\beta s(\xi-\xi_B)\right)\\
        &+e_W b e^{-\beta(\xi-\xi_B)^2}\left(\varpi-b e^{-\beta(\xi-\xi_B)^2}\right)\\
        \geq& be^{-\beta(\xi-\xi_B)^2}\left( -2d_B\beta-\frac{s^2}{4d_B}\right)+e_W b e^{-\beta(\xi-\xi_B)^2}\left(\varpi-b e^{-\beta(\xi-\xi_B)^2}\right)\\
        \geq& be^{-\beta(\xi-\xi_B)^2}\left( -2d_B\beta-\frac{s^2}{4d_B}+e_W(\varpi-b)\right)\\
        \geq & 0.
    \end{aligned}
\end{equation}
Here the first inequality holds true because 
\begin{align*}
    d_B(4\beta^2(\xi-\xi_B)^2-2\beta)+2\beta s(\xi-\xi_B)
\end{align*}
attains its minimum when 
\begin{align} \label{xi min}
    \xi=\mathrm{max}\{\xi_B-\frac{s}{4d_B\beta},\xi_W\}=\xi_B-\frac{s}{4d_B \beta}.
\end{align}
The final inequality holds, as provided by \eqref{2.17 condition}. We pause to remark that for a fixed $s$, $d_B$ and $\beta$, one may choose large enough $\xi_B$ such that \eqref{xi min} holds.

Next, let us address the final case $\xi\leq \xi_W$. We could adopt a similar calculation to obtain 
\begin{equation}
    \begin{aligned}
        d_B& \underline{B}_{\xi\xi}-s\underline{B}_{\xi}+e_W \underline{B}(\underline{W}-\underline{B})\\
        =&be^{-\beta(\xi-\xi_B)^2}\left(d_B(4\beta^2(\xi-\xi_B)^2-2\beta)+2\beta s(\xi-\xi_B)\right)\\
        &+e_W b e^{-\beta(\xi-\xi_B)^2}\left(e^{\lambda_W \xi}-\omega e^{\delta_W\lambda_W \xi}-b e^{-\beta(\xi-\xi_B)^2}\right)\\
        \geq &be^{-\beta(\xi-\xi_B)^2}\Big(d_B(4\beta^2(\xi_W-\xi_B)^2-2\beta)+2\beta s(\xi_W-\xi_B)\\
        &+e_W(1-\omega-b)\Big)\\
        \geq& 0.
    \end{aligned}
\end{equation}
Here we choose sufficiently large $\xi_B$ to assert the last inequality. 

Following a similar idea, one may show that \eqref{upper H}, \eqref{lower H}, \eqref{upper G}, and \eqref{lower G} hold. This completes the proof of Lemma \ref{Lemma UL WB}.
\end{proof}

\subsection{Asymptotic behaviors}
Let $(W,H,G,B)$ denote a solution of \eqref{s W}, whose existence will be given later. We are interested in its asymptotic behavior. Our strategy is based on the shrinking box argument. For more details of this kind argument, please see, for instance, \cite{chen2017traveling,ChenHsiaoWang2024}.

\begin{lemma} \label{shrinking box lemma}
    Recall \eqref{overline W}-\eqref{underline B}. Assume that the solution $(W,H,G,B)$ of the system \eqref{s W} satisfies $(\underline{W},\underline{H},\underline{G},\underline{B})\leq (W,H,G,B) \leq (\overline{W},\overline{H},\overline{G},\overline{B})$. Then the solution has asymptotic behaviors \eqref{limit left} and \eqref{limit right}.
\end{lemma}
\begin{proof} It is straightforward to see that
\begin{align*}
    \lim_{\xi\rightarrow -\infty} (\underline{W},\underline{H},\underline{G},\underline{B})=\lim_{\xi\rightarrow -\infty} (\overline{W},\overline{H},\overline{G},\overline{B})=(0,0,0,0).
\end{align*}
Then, by the squeeze theorem, we obtain that
\begin{align*}
    \lim_{\xi\rightarrow -\infty}(W,H,G,B)=(0,0,0,0).
\end{align*}
Let us denote
\begin{align}
\limsup_{\xi\rightarrow +\infty} ~(W,H,G,B)=(W^+,H^+,G^+,B^+),
\end{align}
and 
\begin{align}
\liminf_{\xi\rightarrow +\infty} ~(W,H,G,B)=(W^-,H^-,G^-,B^-).
\end{align}
We first focus on the tail behaviors of $(W,B)$. To show that 
\begin{align*}
(W^\pm,H^\pm,G^\pm,B^\pm)=(W^*,H^*,G^*,B^*)=(W^*,H^*,H^*,W^*),    
\end{align*} 
consider the auxiliary functions
\begin{align*}
    &m_W(\theta)=\theta W^*,~M_W(\theta)=\theta W^*+(1-\theta)(1+\epsilon),\\
    &m_H(\theta)=\theta H^*,~M_H(\theta)=\theta H^*+(1-\theta)(1+\epsilon),
\end{align*} where
$\theta \in [0,1]$ and $\epsilon$ is a sufficiently small constant such that 
\begin{equation} \label{epsilon condition}
    \begin{aligned}
      &r_W-r_W\kappa_1(1+\epsilon)-c_W(1+\epsilon)>0,\\
      &r_H-r_H\kappa_2(1+\epsilon)-c_H(1+\epsilon)>0.
    \end{aligned}
\end{equation}
We observe that $m_W$ and $m_H$ are strictly increasing, $M_W$ and $M_H$ are strictly decreasing and, for $0<\theta_1<\theta_2<1$, 
\begin{align*}
    0=m_W(0)<m_W(\theta_1)<m_W(\theta_2)<m_W(1)=W^*,
\end{align*}
\begin{align*}
    W^*=M_W(1)<M_W(\theta_2)<M_W(\theta_1)<M_W(0)=1+\epsilon,
\end{align*}
and, similarly, 
\begin{align*}
    0=m_H(0)<m_H(\theta_1)<m_H(\theta_2)<m_H(1)=H^*,
\end{align*}
\begin{align*}
    H^*=M_H(1)<M_H(\theta_2)<M_H(\theta_1)<M_H(0)=1+\epsilon.
\end{align*}

Let
\begin{align*}
        \Theta=\{\theta: \theta\in [0,1) ~\forall~W^{\pm}, B^{\pm}\in (m_W(\theta),M_W(\theta))~~\mbox{and}~~H^{\pm}, G^{\pm}\in (m_H(\theta),M_H(\theta))\}.
\end{align*}
Recall \eqref{overline W}-\eqref{underline B}, so that it is not hard to see that $\Theta$ is not an empty set, since $0\in\Theta$. We claim that $\theta^*=\sup(\Theta)=1$. Suppose, on the contrary, that $\theta^*<1$. By continuity argument, we know 
\begin{align*}
    W^\pm,B^\pm\in \left[m_W(\theta^*),M_W(\theta^*)\right]~~\mbox{and}~~H^\pm,G^\pm\in \left[m_H(\theta^*),M_H(\theta^*)\right]
\end{align*}
and at least one of the following cases holds 
\begin{align*}
    (i)  &~W^-=m_W(\theta^*),~~~&(ii)~B^-=m_W(\theta^*),\\
    (iii)&~W^+=M_W(\theta^*),~~~&(iv)~B^+=M_W(\theta^*),\\
    (v)  &~H^-=m_H(\theta^*),~~~&(vi)~G^-=m_H(\theta^*),\\
    (vii)&~H^+=M_H(\theta^*),~~~&(viii)~G^+=M_H(\theta^*).
\end{align*}

Let us assume $(i)$ holds. Recalling \eqref{epsilon condition}, a straightforward calculation reveals that
\begin{equation} \label{liminf >0}
    \begin{aligned}
    &\liminf_{\xi\rightarrow +\infty} \left( r_W W(\xi)(1-W(\xi)-\kappa_1 H(\xi))-c_W B(\xi)W(\xi)\right) \\
    &\geq r_W m_W(\theta^*)(1-m_W(\theta^*))-(r_W\kappa_1 M_H(\theta^*)+c_W M_W(\theta^*))m_W(\theta^*)\\
    &= r_W\theta^*W^*(1-\theta^*)\left(1-\kappa_1(1+\epsilon)-\frac{c_W}{r_W}\left(1+\epsilon\right)\right)\\
    &>0.
\end{aligned} 
\end{equation}

Under this assumption, we further divide the proof into three cases. In the first case, suppose there is a $\xi_0$ such that $W_{\xi}\geq 0$ when $\xi>\xi_0$, i.e. $W$ is eventually monotonic increasing. Then \eqref{overline W} guarantees that $W(+\infty)$ exists and, hence, $\int_0^\infty W_{\xi}(\xi) d\xi=W(+\infty)-W(0)$ is finite. This implies that $\liminf_{\xi\rightarrow +\infty} W_{\xi}(\xi)=0$. Consider a sequence $\{\xi_n\}_{n=1}^{\infty}$ tending to positive infinity such that 
\begin{align*}
    \lim_{n\rightarrow \infty} W(\xi_n)=m(\theta^*)~\mbox{and}~\lim_{n\rightarrow \infty} W_{\xi}(\xi_n)=0.
\end{align*}

Integrating the first equation in \eqref{s W} over $[0,\xi_n]$, one finds that
\begin{align*}
    d_W&(W_{\xi}(0)-W_{\xi}(\xi_n))-s(W(0)-W(\xi_n))\\
    =& \int_0^{\xi_n} r_W W(\xi)(1-W(\xi)-\kappa_1 H(\xi))-c_W B(\xi)W(\xi) d\xi.
\end{align*}
It is clear that 
\begin{align*}
    \lim_{n\rightarrow +\infty} &\left(d_W(W_{\xi}(0)-W_{\xi}(\xi_n))-s(W(0)-W(\xi_n))\right)\\
    &=d_W W_{\xi}(0)-s(W(0)-\theta^*W^*),
\end{align*}
which is bounded.
On the other hand, \eqref{liminf >0} yields
\begin{align*}
\lim_{n\rightarrow \infty}&\int_0^{\xi_n} r_W W(\xi)(1-W(\xi)-\kappa_1 H(\xi))-c_W B(\xi)W(\xi) d\xi\\
    &\geq \lim_{n\rightarrow \infty}\xi_n \left(  r_W\theta^*W^*(1-\theta^*)\left(1-\kappa_1(1+\epsilon)-\frac{c_W}{r_W}\left(1+\epsilon\right)\right) \right)=\infty.
\end{align*}
We get a contradiction. For the second case, suppose there is a $\xi_0$ such that $W_{\xi}\leq 0$ when $\xi>\xi_0$. We still get a contradiction by following a similar manner above.

For the third case ($W$ is oscillatory at $+\infty$), one may choose a sequence $\{\xi_n\}_{n=1}^{\infty}$ tending to positive infinity such that $\xi_n$ is a local minimizer of $W$ for each $n\in \mathbb{N}$. Notice that this implies $W_{\xi}(\xi_n)=0$ and $W_{\xi\xi}(\xi_n)\geq 0$ for each $n\in\mathbb{N}$. Inequality \eqref{liminf >0} then yields that
\begin{align*}
    &\liminf_{n\rightarrow \infty}\Big(d_W W_{\xi\xi}(\xi_n)-sW_{\xi}(\xi_n)\\
    &~~+r_W W(\xi_n)(1-W(\xi_n)-\kappa_1 H(\xi_n))-c_W B(\xi_n)W(\xi_n)\Big)>0.
\end{align*}
A contradiction then asserts that $W^-=m(\theta^*)$ cannot hold.

Let us assume $(ii)$ holds. Based on the previous argument, $W^->m(\theta^*)=B^-$. A straightforward computation reveals that 
\begin{align*}
    \liminf_{\xi\rightarrow+\infty}(e_W B(\xi)(W(\xi)-B(\xi))= e_W B^-(W^--B^-)>0.
\end{align*}
We analyze $B$ by considering two different cases, $B$ is eventually monotone and $B$ is oscillatory at $+\infty$. Then, a similar argument as the previous one asserts that $(ii)$ cannot hold.

One may similarly show that $(iii)$ and $(iv)$ cannot hold. Also, using \eqref{epsilon condition} and following the same argument, we obtain $(v), (vi), (vii)$ and $(viii)$ cannot hold. Consequently,
we must have $\theta^*=\sup(\Theta)=1$. In other words, 
\begin{align} \label{limit for W}
    \lim_{\xi\rightarrow+\infty}(W(\xi),H(\xi),G(\xi),B(\xi))=(W^*,H^*,H^*,W^*)=(W^*,H^*,G^*,B^*).
\end{align}
This proof is now complete.
\end{proof}

\section{Existence of traveling wave solutions} \label{section 4}

\subsection{Function spaces} \label{Function spaces}
The function spaces we use are standard (see, \cite{ma2001traveling,chen2017traveling,ChenHsiaoWang2024}, for instance). For $\sigma>0$, we consider
\begin{align} \label{X sigma}
      X_{\sigma}=\{\phi(\xi)=(\phi_1(\xi),\phi_2(\xi),\phi_3(\xi),\phi_4(\xi))\in C(\mathbb{R},\mathbb{R}^4): \|\phi\|_{\sigma}<\infty\},
\end{align}
with the weighted norm 
\begin{align*} 
\|\phi\|_{\sigma}:=\sup\limits_{\xi\in\mathbb{R}}\left|\phi(\xi) e^{-\sigma|\xi|}\right|. 
\end{align*}
Let
\begin{equation} \label{convex set gamma 2}
    \Gamma = \left\{ \phi\in X_\sigma \,:\, 
\begin{aligned}
&\underline{W}(\xi)\leq\phi_1(\xi)\leq \overline{W}(\xi),~ \underline{H}(\xi)\leq\phi_2(\xi)\leq \overline{H}(\xi)\\
& \underline{G}(\xi)\leq\phi_3(\xi)\leq \overline{G}(\xi),~ ~ \underline{B}(\xi)\leq\phi_4(\xi)\leq \overline{B}(\xi)
\end{aligned}
\right\}.
\end{equation}

We pause to remark that $(X_\sigma,\|\cdot\|_\sigma)$ is a Banach space and $\Gamma$ is a nonempty closed convex subset of $(X_\sigma,\|\cdot\|_\sigma)$.

Define the function $F(y_1,y_2,y_3,y_4)=(F_1,F_2,F_3,F_4)$ as follows:
\begin{align*}
    F_1&:=\alpha y_1+r_W y_1(1-y_1-\kappa_1 y_2)-c_W y_4y_1, \\
    F_2&:=\alpha y_2+r_H y_2(1-y_2-\kappa_2 y_1)-c_H y_3 y_2,\\
    F_3&:=\alpha y_3+e_H y_3(y_2-y_3),\\
    F_4&:=\alpha y_4+e_W y_4(y_1-y_4),
\end{align*}
where $\alpha$ is a sufficiently large constant. A straightforward calculation reveals that by choosing 
\begin{align} \label{Alpha}
    \alpha>\max\{r_W(1+\kappa_1)+c_W,r_H(1+\kappa_2)+c_H,e_H,e_W\}
\end{align} 
we obtain
\begin{align} \label{condi 1}
    \frac{\partial F_i}{\partial y_i}>0,~\mbox{in}~0\leq y_i \leq 1
\end{align}
for $i=1,2,3,4$.
Also, we have
\begin{align} \label{condi 2}
    \frac{\partial F_1}{\partial y_2} \leq 0,~
    \frac{\partial F_1}{\partial y_4} \leq 0,~ \frac{\partial F_2}{\partial y_1} \leq 0,~
    \frac{\partial F_2}{\partial y_3} \leq 0,~
    \frac{\partial F_3}{\partial y_2} \geq 0,~ \frac{\partial F_4}{\partial y_1} \geq 0,
\end{align}
for $0\leq y_i \leq 1$ and $i=1,2,3,4$.

We rewrite \eqref{s W} as follows:
\begin{equation} \label{s W F form}
    \begin{aligned} 
    d_1 W_{\xi\xi}&-sW_{\xi}-\alpha W+F_1(W,H,B)=0,\\ 
    d_2 H_{\xi\xi}&-sH_{\xi}-\alpha H+F_2(W,H,G)=0,\\ 
    d_3 G_{\xi\xi}&-sG_{\xi}-\alpha G+F_3(H,G)=0,\\ 
    d_4 B_{\xi\xi}&-sB_{\xi}-\alpha B+F_4(W,B)=0,
\end{aligned}
\end{equation}
where $d_1=d_W$, $d_2=d_H$, $d_3=d_G$ and $d_4=d_B$ for notational simplicity. Also, we define 
\begin{align} \label{root of characteristic polynomial}
    \lambda_i^\pm=\frac{s\pm \sqrt{s^2+4 \alpha d_i}}{2 d_i},\,\,i=1,2,3,4.
\end{align}
We see that $\lambda_i^-<0<\lambda_i^+$ and 
\begin{align} \label{polynimial lambda}
    d_i {\lambda_i^\pm}^2-s\lambda_i^\pm-\alpha=0,\,\,i=1,2,3,4.
\end{align}

\subsection{Schauder fixed point theorem}

Our goal here is to prove the existence of \eqref{s W F form}. To this end, we apply the Schauder fixed point theorem to show that there exists $(W,H,G,B)\in C^{\infty}(\mathbb{R},\mathbb{R}^4)$ that satisfies \eqref{s W F form} or, equivalently, \eqref{s W} with asymptotic behaviors \eqref{limit left} and \eqref{limit right}. 
\begin{lemma} [Strong version of Schauder's Theorem \cite{morris1975schauder,fonseca1995degree,zeidler2012applied,zeidler1986}] \label{Schauder thm}
    Let $\Gamma$ be a nonempty closed convex subset of a Banach space and $P$ a continuous map from $\Gamma$ into a compact subset of $\Gamma$. Then $P$ has a fixed point. 
\end{lemma}

Let us prove Theorem \ref{main 2.1} as follows.
\begin{proof}
    Recalling \eqref{X sigma}, \eqref{root of characteristic polynomial} and \eqref{convex set gamma 2}, let 
    \begin{align} \label{sigma condition}
    \sigma\in (0,\min\{|\lambda_1^-|,|\lambda_2^-|,|\lambda_3^-|,|\lambda_4^-|\}).    
    \end{align}
    For $(W,H,G,B)\in X_\sigma$, we consider the operator $P=(P_1,P_2,P_3,P_4): \Gamma\rightarrow X_\sigma$ defined as follows:
 \begin{align} \label{P1}
        P_1(W,H,G,B)(\xi)=\frac{1}{d_1(\lambda_1^+-\lambda_1^-)}\left[\int_{-\infty}^{\xi} e^{\lambda_1^-(\xi-z)}+\int_{\xi}^{\infty} e^{\lambda_1^+(\xi-z)} \right] F_1(W,H,B)(z)dz,
    \end{align}
    
    \begin{align} \label{P2}
        P_2(W,H,G,B)(\xi)=\frac{1}{d_2(\lambda_2^+-\lambda_2^-)}\left[\int_{-\infty}^{\xi} e^{\lambda_2^-(\xi-z)}+\int_{\xi}^{\infty} e^{\lambda_2^+(\xi-z)} \right] F_2(W,H,G)(z)dz,
    \end{align}
    
    \begin{align} \label{P3}
        P_3(W,H,G,B)(\xi)=\frac{1}{d_3(\lambda_3^+-\lambda_3^-)}\left[\int_{-\infty}^{\xi} e^{\lambda_3^-(\xi-z)}+\int_{\xi}^{\infty} e^{\lambda_3^+(\xi-z)} \right] F_3(H,G)(z)dz,
    \end{align} 
    and
     \begin{align} \label{P4}
        P_4(W,H,G,B)(\xi)=\frac{1}{d_4(\lambda_4^+-\lambda_4^-)}\left[\int_{-\infty}^{\xi} e^{\lambda_4^-(\xi-z)}+\int_{\xi}^{\infty} e^{\lambda_4^+(\xi-z)} \right] F_4(W,B)(z)dz.
    \end{align} 
Note that $P_1$ is independent on $G$, $P_2$ is independent on $B$, $P_3$ is independent on $W$ and $B$, and $P_4$ is independent on $H$ and $G$. The goal here is to show $P$ has a fixed point in $\Gamma$. We first show that $P$ maps $\Gamma$ into $\Gamma$ itself. Recalling \eqref{condi 1} and \eqref{condi 2}, one may observe that
\begin{align*}
P_1(\underline{W},\overline{H},\overline{B})(\xi)\leq P_1(W,H,B)(\xi) \leq P_1(\overline{W},\underline{H},\underline{B})(\xi),   
\end{align*}

\begin{align*}
P_2(\overline{W},\underline{H},\overline{G})(\xi)\leq P_2(W,H,G)(\xi) \leq P_2(\underline{W},\overline{H},\underline{G})(\xi),   
\end{align*}

\begin{align*}
P_3(\underline{H},\underline{G})(\xi)\leq P_3(H,G)(\xi) \leq P_3(\overline{H},\overline{G})(\xi),  
\end{align*}
and
\begin{align*}
P_4(\underline{W},\underline{B})(\xi)\leq P_4(W,B)(\xi) \leq P_4(\overline{W},\overline{B})(\xi),  
\end{align*}
for $\xi\in\mathbb{R}$. Then by
an elementary but tedious computation (see Appendix \ref{Proof P maps Gamma into itself}), we have
\begin{align} \label{W<P}
\underline{W}(\xi) \leq P_1(\underline{W},\overline{H},\overline{B})(\xi)\leq P_1(\overline{W},\underline{H},\underline{B})(\xi)\leq \overline{W}(\xi),   
\end{align}

\begin{align} \label{H<P}
\underline{H}(\xi) \leq P_2(\overline{W},\underline{H},\overline{G})(\xi)\leq P_2(\underline{W},\overline{H},\underline{G})(\xi)\leq \overline{H}(\xi),   
\end{align}

\begin{align} \label{G<P}
\underline{G}(\xi) \leq P_3(\underline{H},\underline{G})(\xi)\leq P_3(\overline{H},\overline{G})(\xi)\leq \overline{G}(\xi),  
\end{align}
and
\begin{align} \label{B<P}
\underline{B}(\xi) \leq P_4(\underline{W},\underline{B})(\xi)\leq P_4(\overline{W},\overline{B})(\xi)\leq \overline{B}(\xi),  
\end{align}
for $\xi\in\mathbb{R}$. This implies that $P(\Gamma)\subset \Gamma$.

Second, recalling \eqref{sigma condition}, we show $P:\Gamma\rightarrow\Gamma$ is continuous with respect to the norem $\|\cdot\|_\sigma$ in $X_\sigma$. The proof of continuity of $P$ is standard. However, we provide the proof in Appendix \ref{continuity of P} for completeness. Our next goal is to show $P(\Gamma)$ is inside a compact subset of $\Gamma$ and apply Lemma \ref{Schauder thm}. In fact, we prove slightly more. We actually prove $P$ is a compact operator.

We observe that for all $\phi\in\Gamma$,
\begin{align} \label{uniform bdd}
    \|P(\phi)\|_\sigma\leq\|(\overline{W},\overline{H},\overline{G},\overline{B})\|_\sigma\leq 2,
\end{align}
so $P$ is bounded. Also, one may observe that
\begin{align*}
    |F_1(W,H,B)|&=|\alpha W+r_W W(1-W-\kappa_1 H)-c_W BW |\leq \alpha+r_W(2+\kappa_1)+c_W,\\
    |F_2(W,H,G)|&=|\alpha H+r_H H(1-H-\kappa_2 W)-c_H GH |\leq \alpha+r_H(2+\kappa_2)+c_H,\\
    |F_3(H,G)|&=|\alpha G+e_H G(H-G)|\leq\alpha+e_H,\\
    |F_4(W,B)|&=|\alpha B+e_W B(W-B)|\leq\alpha+e_W,
\end{align*}
for any $(W,H,G,B)\in\Gamma$. Therefore, we obtain
\begin{equation*} 
\begin{aligned}
    &\left|\frac{d}{d\xi}P_1(W,H,B)(\xi)\right|\\
    &\leq \frac{-\lambda_1^-}{d_1(\lambda_1^+-\lambda_1^-)}\int_{-\infty}^{\xi} e^{\lambda_1^-(\xi-z)}\left|F_1(W,H,B)(z)\right|dz\\
    &\,\,\,\,+\frac{\lambda_1^+}{d_1(\lambda_1^+-\lambda_1^-)}\int_{\xi}^{\infty} e^{\lambda_1^+(\xi-z)}\left| F_1(W,H,B)(z)\right|dz\\
    &\leq \frac{2(\alpha+r_W(2+\kappa_1)+c_W)}{d_1(\lambda_1^+-\lambda_1^-)}.
\end{aligned}    
\end{equation*}
Similarly, we obtain
\begin{align*} 
    \left|\frac{d}{d\xi}P_2(W,H,G)(\xi)\right|\leq \frac{2(\alpha+r_H(2+\kappa_2)+c_H)}{d_2(\lambda_2^+-\lambda_2^-)},
\end{align*}

\begin{align*} 
    \left|\frac{d}{d\xi}P_3(H,G)(\xi)\right|\leq \frac{2(\alpha+e_H)}{d_3(\lambda_3^+-\lambda_3^-)}.
\end{align*}
and
\begin{align*} 
    \left|\frac{d}{d\xi}P_4(W,B)(\xi)\right|\leq \frac{2(\alpha+e_W)}{d_4(\lambda_4^+-\lambda_4^-)}.
\end{align*}
We know the derivative of $P$ is bounded. This implies that there exists a $L>0$ such that
\begin{align*}
    |P(\phi)(\xi)-P(\phi)(z)|\leq L|\xi-z|.
\end{align*}
In other words, $P(\Gamma)$ is equicontinuous, so for each $\epsilon>0$, there is a $\delta>0$ such that 
\begin{align} \label{equi conti}
    |\xi-z|<\delta~~\mbox{and}~~\phi\in \Gamma~~\mbox{impliy}~~|P(\phi)(\xi)-P(\phi)(z)|<\epsilon.
\end{align}

Our next goal is to show $P:\Gamma\rightarrow\Gamma$ is compact. Observe that this proof could be well-known, but we prove it in Appendix \ref{P is compact} for completeness. Applying Lemma \ref{Schauder thm}, $P$ has a fixed point $(W,H,G,B)$ and, hence, recalling \eqref{P2} and \eqref{P3}, $(W,H,G,B)\in C^{\infty}(\mathbb{R},\mathbb{R}^4)$. Also, a straightforward calculation reveals that 
 \begin{equation*}
        \begin{aligned}
         d_1 (P_1(W,H,B))_{\xi\xi}&-s(P_1(W,H,B))_{\xi}-\alpha P_1(W,H,B)+F_1(W,H,B)=0,\\
          d_2 (P_2(W,H,G))_{\xi\xi}&-s(P_2(W,H,G))_{\xi}-\alpha P_2(W,H,G)+F_2(W,H,G)=0,\\
         d_3 (P_3(H,G))_{\xi\xi}&-s(P_3(H,G))_{\xi}-\alpha P_3(H,G)+F_3(H,G)=0,\\
         d_4 (P_4(W,B))_{\xi\xi}&-s(P_4(W,B))_{\xi}-\alpha P_4(W,B)+F_4(W,B)=0,
    \end{aligned}
    \end{equation*}
so $(W,H,G,B)$ satisfies \eqref{s W}. Finally, by Lemma \ref{shrinking box lemma}, \eqref{limit left} and \eqref{limit right} hold. 

The proof is now complete. 
\end{proof}

\section{Traveling waves connecting $(0,H_*,G_*,0)$ and $(W^*,H^*,G^*,B^*)$} \label{sec 5}

\subsection{Existence}
In this subsection, we prove the existence of traveling wave solutions of \eqref{s W} satisfying the asymptotic behaviors \eqref{coe to coe}  and \eqref{limit right}. 

Let us construct another upper and lower solutions for system \eqref{s W}. 

\begin{lemma} \label{Lemma coe to coe}
    Suppose that the assumptions described in Theorem \ref{main 2.2} hold. For $s_m<s<s_{M_1}$, if \eqref{1-cw/rw}-\eqref{1-w-ch/rh} hold true, recalling \eqref{overline W}, \eqref{overline B}, \eqref{underline W}, and \eqref{underline B}, then the functions $(\overline{W},1,1,\overline{B})$ and $(\underline{W},h,h,\underline{B})$ are upper and lower solutions, respectively, satisfying Definition \ref{UL for WB}, where $\xi_W<0<\xi_B$ for some $\xi_W$ and some large enough $\xi_B$.
\end{lemma}
\begin{proof}
    The proof is similar to the previous one, so we omit it.
\end{proof}

Let us prove Theorem \ref{main 2.2} here.
\begin{proof}
    Similarly, the shrinking box argument implies that
    \begin{align*} 
    \lim_{\xi\rightarrow -\infty}(W,H,G,B)=(0,H_*,G_*,0)~~\mbox{and}~~\lim_{\xi\rightarrow\infty}(W,H,G,B)=(W^*,H^*,G^*,B^*).
\end{align*}
The proof of Theorem \ref{main 2.2} is complete.
\end{proof}

\subsection{N-barrier maximum principle}

In this subsection, we estimate the total value of both wood and grass in this ecological system. In other words, given $c_1,c_2>0$, can we find non-trivial lower bounds (depending on the parameters) of $c_1 W(\xi)+c_2 H(\xi)$, where $W(\xi)$ and $H(\xi)$ satisfy \eqref{s W} and \eqref{coe to coe}? To obtain such estimates, we successfully generalize a so called $N-$barrier maximum (minimum) principle, which has been applied to Lotka-Volterra systems in several occasions, to our complex waves. For more details about this maximum principle,  we refer to \cite{hsiao2022estimates} Theorem 3, \cite{chen2020discrete} Theorem 1.2, \cite{hung2016n} Theorem 1.1, \cite{chen2016nonexistence} Theorem 2.2, \cite{chen2016n} Theorem 1.1, and \cite{chen2016maximum} Theorem 1.1.

Given two arbitrary positive constants $c_1$ and $c_2$, recall \eqref{q p}. Let us define a useful function as follows:
\begin{align*} 
    F(W,H,G,B;c_1,c_2):=&\frac{c_1}{d_W} r_W W(1-W-\kappa_1 H-\frac{c_W}{r_W}B)\\
    &+\frac{c_2}{d_H} r_H H(1-H-\kappa_2 W-\frac{c_H}{r_H}G).
\end{align*}

\begin{lemma}
The quadratic curve $F(W,H,1,1;c_1,c_2)=0$ is an ellipse if $\triangle<0$; a hyperbola if $\triangle>0$; a parabola if $\triangle=0$, where 
\begin{align} \label{triangle}
    \triangle=\left(\frac{c_1 r_W \kappa_1}{d_W}+\frac{c_2 r_H \kappa_2}{d_H}\right)^2-\frac{4c_1 c_2 r_W r_H}{d_W d_H}.
\end{align}
\end{lemma}
\begin{proof}
    The discriminant of the quadratic curve is \eqref{triangle}. This completes the proof.
\end{proof}

We pause to remark that $(0,H_*)$ and $(W^*,H^*)$ are not in the region $E=\{(W,H)\in \mathbb{R}^2:W\geq 0, H\geq 0, F(W,H,1,1;c_1,c_2)\geq 0\}$, since $F(0,H_*,1,1;c_1,c_2)<0$ and $F(W^*,H^*,1,1;c_1,c_2)<0$. Also, recall \eqref{tilde W H}. We notice that all line segments $L_1=\{t(\tilde{W},0)+(1-t)(0,\tilde{H}):t\in(0,1)\}$, $L_2=\{(W,0):0\leq W\leq \tilde{W}\}$, $L_3=\{(0,H):0\leq H\leq\tilde{H}\}$ are contained in $E$.

Now, we are ready to prove Theorem \ref{main thm 3}.
\begin{proof}
 Let us classify four cases as follows:
 \begin{enumerate}[label=(\roman*)]
    \item When $c_1 \tilde{W}\geq c_2 \tilde{H}$, $d_H <d_W$, we choose $(l_1,l_2,l_3)=(c_2 \tilde{H},\frac{c_2}{d_W}\tilde{H},\frac{c_2 d_H}{d_W}\tilde{H})$;
    \item When $c_1 \tilde{W}\geq c_2 \tilde{H}$, $d_H\geq d_W$, we choose $(l_1,l_2,l_3)=(c_2 \tilde{H},\frac{c_2}{d_H}\tilde{H},\frac{c_2 d_W}{d_H}\tilde{H})$;
    \item When $c_1 \tilde{W}<c_2 \tilde{H}$, $d_H<d_W$, we choose $(l_1,l_2,l_3)=(c_1 \tilde{W},\frac{c_1}{d_W}\tilde{W},\frac{c_1 d_H}{d_W}\tilde{W})$;
    \item When $c_1 \tilde{W}<c_2 \tilde{H}$, $d_H\geq d_W$, we choose $(l_1,l_2,l_3)=(c_1 \tilde{W},\frac{c_1}{d_H}\tilde{W},\frac{c_1 d_W}{d_H}\tilde{W})$.
\end{enumerate}
With these preparations, recalling \eqref{q p}, we have
\begin{align*}
    Q_{l_1}\subset P_{l_2}\subset Q_{l_3}\subset E,
\end{align*}
where
\begin{equation} \label{QPQ}
\begin{aligned}
    Q_{l_1}&=\left\{(W,H): q\leq l_1, W\geq 0, H\geq 0\right\},\\
    P_{l_2}&=\left\{(W,H): p\leq l_2, W\geq 0, H\geq 0\right\},\\
    Q_{l_3}&=\left\{(W,H): q\leq l_3, W\geq 0, H\geq 0\right\},
\end{aligned}
\end{equation}
are called a $N$-barrier in the $WH$ plane for our system. 

Recall \eqref{q p}. Multiplying the first and second equations by $\frac{c_1}{d_W}$ and $\frac{c_2}{d_H}$ respectively, and subsequently summing the results, yields
\begin{align} \label{q p equation}
    q''(\xi)-s p'(\xi)+F(W(\xi),H(\xi),G(\xi),B(\xi);c_1,c_2)=0,~~\xi\in\mathbb{R}.
\end{align}
Suppose the conclusion \eqref{minimum of p} does not hold. There exits $\xi^*\in\mathbb{R}$ such that $q(\xi^*)<l_3$. Since $(W,H)\in C^2(\mathbb{R},\mathbb{R}^2)$ and satisfies \eqref{coe to coe}, one may assume that 
\begin{align} \label{min q}
    \min\limits_{\xi\in\mathbb{R}}q(\xi)=q(\xi^*).
\end{align}

We denote $\overline{\xi}$ the first point at which the solution $(W(\xi),H(\xi))$ intersects the line $q=l_3$ in the $WH-$plane when $\xi$ moves from $\xi^*$ towards $-\infty$. We integrate \eqref{q p equation} with respect to $\xi$ from $\overline{\xi}$ to $\xi^*$ to obtain
\begin{align} \label{q p int eq}
    q'(\xi^*)-q'(\overline{\xi})-s(p(\xi^*)-p(\overline{\xi}))+\int_{\overline{\xi}}^{\xi^*} F(W,H,G,B;c_1,c_2)(\xi) d\xi=0.
\end{align}
On the other hand, recalling \eqref{min q}, we obtain
\begin{align} \label{q' estimate}
     q'(\xi^*)-q'(\overline{\xi})=-q'(\overline{\xi})\geq 0.
\end{align}
Recalling \eqref{QPQ}, we arrive at
\begin{align} \label{p estimate}
     s(p(\xi^*)-p(\overline{\xi}))<0.
\end{align}
Finally, we have $(W(\xi),H(\xi))\in E$ and $B(\xi)< 1$, $G(\xi)<1$ when $\overline{\xi}\leq\xi\leq\xi^*$, so 
\begin{align} \label{F estimate}
    F(W,H,G,B;c_1,c_2)(\xi)\geq F(W,H,1,1;c_1,c_2)(\xi)\geq 0.
\end{align}
Combining \eqref{q' estimate},\eqref{p estimate}, and \eqref{F estimate}, one may find 
\begin{align*}
    q'(\xi^*)-q'(\overline{\xi})-s(p(\xi^*)-p(\overline{\xi}))+\int_{\overline{\xi}}^{\xi^*} F(W,H,G,B;c_1,c_2)(\xi) d\xi>0.
\end{align*}
This contradicts to \eqref{q p int eq}. The proof of Theorem \ref{main thm 3} is now complete.
\end{proof}

\section{Acknowledgement}
We are deeply grateful to Ching-Lin Huang for his insights on ecology, particularly his valuable suggestions regarding ecological stability. Chiun-Chuan Chen
is supported by the National Science and Technology Council, Taiwan (Grant Number
111-2115-M-002-008-MY3) and the National Center for Theoretical Sciences, Taiwan
(NCTS).

\section{Appendix}
\subsection{Proof of $P(\Gamma)\subset \Gamma$ } \label{Proof P maps Gamma into itself}

For \eqref{H<P}, let us first show that $P_2(\overline{W},\underline{H},\overline{G})\geq \underline{H}$. Recalling \eqref{underline H}, we first assume $\xi>\xi_H$ and make a straightforward calculation as follows:
\begin{equation*}
    \begin{aligned}
        &P_2(\overline{W},\underline{H},\overline{G})(\xi)\\
        =&\frac{1}{d_H(\lambda^+_2-\lambda^-_2)}\left[\int_{-\infty}^\xi e^{\lambda^-_2(\xi-z)}+\int_{\xi}^\infty e^{\lambda^+_2(\xi-z)}\right] F_2(\overline{W},\underline{H},\overline{G})(z) dz\\
        \overset{\eqref{lower H}}{\geq}&\frac{1}{d_H(\lambda^+_2-\lambda^-_2)}\left[\int_{-\infty}^\xi e^{\lambda^-_2(\xi-z)}+\int_{\xi}^\infty e^{\lambda^+_2(\xi-z)}\right] \left(-d_H\underline{H}_{zz}+s\underline{H}_z+\alpha\underline{H}\right) dz.
    \end{aligned}
\end{equation*}
One may notice that 
\begin{equation} \label{5.1 I}
\begin{aligned}
   \int_{\xi}^\infty e^{\lambda^+_2(\xi-z)} \left(-d_H\underline{H}_{zz}+s\underline{H}_z+\alpha\underline{H}\right) dz=\frac{\alpha}{\lambda_2^+}\underline{H}(\xi)
\end{aligned}    
\end{equation}
and
\begin{equation} \label{5.2 II}
    \begin{aligned}
        \int_{\xi_H}^{\xi}e^{\lambda^-_2(\xi-z)} \left(-d_H\underline{H}_{zz}+s\underline{H}_z+\alpha\underline{H}\right) dz=\frac{\alpha}{\lambda^-_2}e^{\lambda^-_2(\xi-\xi_H)}\underline{H}(\xi_H)-\frac{\alpha}{\lambda_2^-}\underline{H}(\xi).
    \end{aligned}
\end{equation}
On the other hand, we have
\begin{equation} \label{5.3 III}
    \begin{aligned}
        &\int_{-\infty}^{\xi_H} e^{\lambda^-_2(\xi-z)} \left(-d_H\underline{H}_{zz}+s\underline{H}_z+\alpha\underline{H}\right) dz\\
        =&-d_He^{\lambda^-_2(\xi-\xi_H)} \underline{H}_{\xi}(\xi_H^-)-d_H\lambda^-_2\int_{-\infty}^{\xi_H} e^{\lambda^-_2(\xi-z)} \underline{H}_z dz\\
        &+s\int_{-\infty}^{\xi_H} e^{\lambda^-_2(\xi-z)} \underline{H}_z dz-\frac{\alpha}{\lambda^-_2}e^{\lambda^-_2(\xi-\xi_H)}\underline{H}(\xi_H)+\frac{\alpha}{\lambda^-_2}\int_{-\infty}^{\xi_H} e^{\lambda^-_2(\xi-z)}\underline{H}_z dz\\
        \overset{\eqref{polynimial lambda}}{=}&-d_He^{\lambda^-_2(\xi-\xi_H)} \underline{H}_{\xi}(\xi_H^-)-\frac{\alpha}{\lambda^-_2}e^{\lambda^-_2(\xi-\xi_H)}\underline{H}(\xi_H)
    \end{aligned}
\end{equation}
Combining \eqref{5.1 I}, \eqref{5.2 II} and \eqref{5.3 III} and recalling \eqref{underline H}, \eqref{Alpha}, and \eqref{polynimial lambda}, one may obtain that
\begin{align*}
    P_2(\overline{W},\underline{H},\overline{G})(\xi)&\geq \frac{1}{d_H(\lambda^+_2-\lambda^-_2)}\left(\frac{\alpha}{\lambda^+_2}\underline{H}(\xi)-\frac{\alpha}{\lambda^-_2}\underline{H}(\xi)-d_He^{\lambda^-_2(\xi-\xi_H)}\underline{H}_\xi(\xi_H^-)\right)\\
    &\geq \underline{H}(\xi).
\end{align*}
One may follow a similar argument for the case $\xi\leq \xi_H$. Therefore, we show that 
\begin{align} \label{P2>H}
P_2(\overline{W},\underline{H},\overline{G})(\xi)\geq \underline{H}(\xi) ~~\mbox{for}~~\xi\in\mathbb{R}.    
\end{align}

For \eqref{G<P}, let us show $P_3(\underline{H},\underline{G})\geq \underline{G}$. Recall that $\underline{G}\in C^{1}(\mathbb{R},\mathbb{R})$. A straightforward calculation reveals that
\begin{equation} \label{P3>G}
    \begin{aligned}
        &P_3(\underline{H},\underline{G})(\xi)\\
        =&\frac{1}{d_G(\lambda^+_3-\lambda^-_3)}\left[\int_{-\infty}^\xi e^{\lambda^-_3(\xi-z)}+\int_{\xi}^\infty e^{\lambda^+_3(\xi-z)}\right] F_3(\underline{H},\underline{G})(z) dz\\
        \overset{\eqref{lower G}}{\geq}&\frac{1}{d_G(\lambda^+_3-\lambda^-_3)}\left[\int_{-\infty}^\xi e^{\lambda^-_3(\xi-z)}+\int_{\xi}^\infty e^{\lambda^+_3(\xi-z)}\right] \left(-d_G\underline{G}_{zz}+s\underline{G}_z+\alpha\underline{G}\right) dz\\
        =&\underline{G}(\xi).
    \end{aligned}
\end{equation}
Combining \eqref{P2>H} and \eqref{P3>G}, we prove that 
\begin{align*}
    (\underline{H},\underline{G})\leq P(H,G)~~\mbox{for}~~(H,G)\in\Gamma.
\end{align*}
Similarly, one may follow the previous argument to prove that 
\begin{align*}
    (\overline{H},\overline{G})\geq P(H,G)~~\mbox{for}~~(H,G)\in\Gamma.
\end{align*}
Other cases \eqref{W<P} and \eqref{B<P} can be proven similarly. Finally, we obtain $P(\Gamma)\subset\Gamma$.

\subsection{Proof of continuity of $P$ with respect to $\|\cdot\|_\sigma$} \label{continuity of P}
Let $\hat{\phi}=(\hat{W},\hat{H},\hat{G},\hat{B})\in\Gamma$ and $\check{\phi}=(\check{W},\check{H},\check{G},\check{B})\in\Gamma$. Let us assume that $\xi>0$. One may obtain
\begin{equation*}
    \begin{aligned}
        &\left|P_2(\hat{\phi})(\xi)-P_2(\check{\phi})(\xi)\right|\\
        \leq&\frac{1}{d_H(\lambda^+_2-\lambda^-_2)}\left[\int_{-\infty}^\xi e^{\lambda^-_2(\xi-z)}+\int_{\xi}^\infty e^{\lambda^+_2(\xi-z)}\right]\left|F_2(\hat{\phi})(z)-F_2(\check{\phi})(z)\right| dz\\
        \leq& \frac{1}{d_H(\lambda^+_2-\lambda^-_2)}\left[\int_{-\infty}^\xi e^{\lambda^-_2(\xi-z)+\sigma|z|}+\int_{\xi}^\infty e^{\lambda^+_2(\xi-z)+\sigma|z|}\right]\left\|F_2(\hat{\phi})-F_2(\check{\phi})\right\|_\sigma dz\\
        =&\frac{1}{d_H(\lambda^+_2-\lambda^-_2)}\left(\frac{2\sigma}{{\lambda^-_2}^2-\sigma^2}e^{\lambda^-_2 \xi}+\frac{\lambda^+_2-\lambda^-_2}{(-\lambda^-_2+\sigma)(\lambda^+_2-\sigma)}e^{\sigma \xi}\right)\left\|F_2(\hat{\phi})-F_2(\check{\phi})\right\|_\sigma.
    \end{aligned}
\end{equation*}
Similarly, for $\xi\leq 0$, one may find that
\begin{align*}
    &\left|P_2(\hat{\phi})(\xi)-P_2(\check{\phi})(\xi)\right|\\
    \leq&\frac{1}{d_H(\lambda^+_2-\lambda^-_2)}\left(\frac{2\sigma}{{\lambda^+_2}^2-\sigma^2}e^{\lambda^+_2 \xi}+\frac{\lambda^+_2-\lambda^-_2}{(-\lambda^-_2-\sigma)(\lambda^+_2+\sigma)}e^{-\sigma \xi}\right)\left\|F_2(\hat{\phi})-F_2(\check{\phi})\right\|_\sigma.
\end{align*}
Multiplying $e^{-\sigma|\xi|}$ and recalling \eqref{sigma condition}, we obtain
\begin{align} \label{P<F}
    \left|P_2(\hat{\phi})(\xi)-P_2(\check{\phi})(\xi)\right|e^{-\sigma|\xi|} \lesssim \left\|F_2(\hat{\phi})-F_2(\check{\phi})\right\|_\sigma.
\end{align}
Also, recalling \eqref{overline H}, \eqref{overline G}, \eqref{underline H}, \eqref{underline G} and \eqref{convex set gamma 2}, we notice that 
\begin{equation} \label{F<L phi}
    \begin{aligned}
         &\left|F_2(\hat{\phi})(\xi)-F_2(\check{\phi})(\xi)\right|\\
        \leq&\left|(\alpha+r_H)(\hat{H}(\xi)-\check{H}(\xi))\right|+\left|r_H \kappa_2 (\hat{W}(\xi)\hat{H}(\xi)-\check{W}(\xi)\check{H}(\xi))\right|\\
        &+\left|r_H (\hat{H}^2(\xi)-\check{H}^2(\xi))\right|+\left|c_H(\hat{G}(\xi)\hat{H}(\xi)-\check{H}(\xi)\check{G}(\xi))\right|\\
        \lesssim&\left|\hat{\phi}(\xi)-\check{\phi}(\xi)\right|.
    \end{aligned}
\end{equation}
Multiplying $e^{-\sigma|\xi|}$, \eqref{F<L phi} implies that 
\begin{align} \label{F<phi}
    \left|F_2(\hat{\phi})(\xi)-F_2(\check{\phi})(\xi)\right| e^{-\sigma|\xi|}\lesssim \left\|\hat{\phi}-\check{\phi}\right\|_\sigma.
\end{align}
Combining \eqref{P<F} and \eqref{F<phi}, we obtain
\begin{align*}
   \left|P_2(\hat{\phi})(\xi)-P_2(\check{\phi})(\xi)\right| e^{-\sigma|\xi|} \lesssim \left\|\hat{\phi}-\check{\phi}\right\|_\sigma.
\end{align*}
A similar argument can be applied to show 
\begin{align*}
   \left|P_1(\hat{\phi})(\xi)-P_1(\check{\phi})(\xi)\right| e^{-\sigma|\xi|} \lesssim \left\|\hat{\phi}-\check{\phi}\right\|_\sigma,
\end{align*}
\begin{align*}
   \left|P_3(\hat{\phi})(\xi)-P_3(\check{\phi})(\xi)\right| e^{-\sigma|\xi|} \lesssim \left\|\hat{\phi}-\check{\phi}\right\|_\sigma,
\end{align*}
and
\begin{align*}
   \left|P_4(\hat{\phi})(\xi)-P_4(\check{\phi})(\xi)\right| e^{-\sigma|\xi|} \lesssim \left\|\hat{\phi}-\check{\phi}\right\|_\sigma.
\end{align*}
Therefore, we show $P:\Gamma\rightarrow\Gamma$ is continuous with respect to $\|\cdot\|_\sigma$.

\subsection{$P:\Gamma\rightarrow\Gamma$ is compact} \label{P is compact}
We aim to show that $P(\Gamma)$ is relatively compact. Suppose that we are given a sequence of functions $\{\psi_m\}$ in $P(\Gamma)\subset \Gamma$. We pause to remark that the functions 
\begin{align*}
    \psi_m: \mathbb{R}\rightarrow\mathbb{R}^4,~~m=1,2,\ldots
\end{align*}
are continuous and satisfy $\|\psi_m\|_\sigma<\infty$. Also, there are $\{\phi_m\}\subset \Gamma$ such that $P(\phi_m)=\psi_m$, so $\psi_m$ satisfy \eqref{uniform bdd} and \eqref{equi conti} for all $m=1,2,\ldots$.

Let $\mathbb{Q}=\{r_i:i=1,2,\ldots\}$ denote the rational numbers contained in $\mathbb{R}$. By \eqref{uniform bdd}, the sequence $\{\psi_m(r_1)\}$ is bounded in $\mathbb{R}^4$. By Bolzano–Weierstrass theorem, there is a subsequence $\{\psi^{(1)}_m\}$ of $\{\psi_m\}$ such that $\{\psi^{(1)}_m(r_1)\}$ is convergent. In other words, there exists $w_1\in\mathbb{R}^4$ such that 
\begin{align*}
    \psi^{(1)}_m(r_1)\rightarrow w_1~~\mbox{as}~~m\rightarrow\infty
\end{align*}
Then again by $\eqref{uniform bdd}$, the sequence $\{\psi^{(1)}_m(r_2)\}$ is bounded in $\mathbb{R}^4$ and, hence, we apply Bolzano–Weierstrass theorem to obtain a sequence $\{\psi^{(2)}_m\}$ of $\{\psi^{(1)}_m\}$ such that 
\begin{align*}
    \psi^{(2)}_m(r_2)\rightarrow w_2 ~~\mbox{in}~~\mathbb{R}^4~~\mbox{as}~~m\rightarrow\infty.
\end{align*}
Let us continue in this manner \emph{ad infinitum}, that is, constructing a subsequence $\{\psi^{(k)}_m\}$ of $\{\psi_m\}$ such that, as $m\rightarrow\infty$,
\begin{equation}
    \begin{aligned}
        &\psi^{(1)}_1(r_1),\psi^{(1)}_2(r_1),\psi^{(1)}_3(r_1),\ldots~~~~\rightarrow w_1,\\
        &\psi^{(2)}_1(r_2),\psi^{(2)}_2(r_2),\psi^{(2)}_3(r_2),\ldots~~~~\rightarrow w_2,\\
        &\psi^{(3)}_1(r_3),\psi^{(3)}_2(r_3),\psi^{(3)}_3(r_3),\ldots~~~~\rightarrow w_3,\\
        &\ldots
    \end{aligned}
\end{equation}
Here, we denote $\{\psi^{(k+1)}_m\}$ is a subsequence of $\{\psi^{(k)}_m\}$ for all $k=1,2,\ldots$.

Thus we see that the diagonal sequence
\begin{align*}
    \varphi_m:=\psi^{(m)}_m, ~~m=1,2,\ldots
\end{align*}
satisfying
\begin{align} \label{varphi tends to w}
    \varphi_m(r_j)\rightarrow w_j~~\mbox{as}~~m\rightarrow\infty~~\mbox{for}~~\mbox{all}~~j=1,2\ldots.
\end{align}

Recall \eqref{equi conti}. Let $\epsilon>0$ be given. We choose the number $\delta>0$ such that 
\begin{align} \label{varphi-varphi}
    |\varphi_m(\xi)-\varphi_m(z)|<\epsilon, ~~\mbox{as}~~|\xi-z|<\delta.
\end{align}
Since $\varphi_m\in\Gamma$, $|\varphi_m(\xi)|<2$ for all $x\in\mathbb{R}$, $m=1,2,\ldots$. One may choose $n$ that is sufficiently large such that
\begin{align} \label{outside estimate}
    \sup_{z\in[-n,n]^c}\left||\varphi_m(z)-\varphi_l(z)| e^{-\sigma|z|}\right|<\epsilon,
\end{align}
for all $m,l=1,2,\ldots$.

Let us choose a partition $\{\xi_{1,n},\ldots,\xi_{s,n}\}\subset\mathbb{Q}\cap[-n,n]$ such that, for each $\xi\in[-n,n]$, there is some $\xi_{j,n}$ such that
\begin{align} \label{319 shooting}
    |\xi-\xi_{j,n}|<\delta.
\end{align}

By \eqref{varphi tends to w}, for each $j=1,\ldots,s$, the  $\{\varphi_m(\xi_{j,n})\}$ is a convergent sequence, and, hence, it is a Cauchy sequence. There exists a $N(\epsilon,n)>0$ such that 
\begin{align*}
    |\varphi_m(\xi_{j,n})-\varphi_l(\xi_{j,n})|<\epsilon~~\mbox{for}~~\mbox{all}~~m,l\geq N(\epsilon,n), ~j=1,\ldots,s.
\end{align*}
For each $\xi\in[-n,n]$, it follows from \eqref{varphi-varphi} and \eqref{319 shooting} that
\begin{equation} \label{3 epsilon estimate}
    \begin{aligned}
        |\varphi_m(\xi)-\varphi_l(\xi)|\leq& |\varphi_m(\xi)-\varphi_m(\xi_{j,n})|\\
        &+|\varphi_m(\xi_{j,n})-\varphi_l(\xi_{j,n})|+|\varphi_l(\xi_{j,n})-\varphi_l(\xi)|<3\epsilon,
    \end{aligned}
\end{equation}
for all $m,l\geq N(\epsilon,n)$.

Combining \eqref{outside estimate} and \eqref{3 epsilon estimate}, we obtain
\begin{align*}
    \|\varphi_m-\varphi_l\|_\sigma=\sup\limits_{\xi\in\mathbb{R}}\left||\varphi_m(\xi)-\varphi_l(\xi)| e^{-\sigma|\xi|}\right|<3\epsilon,
\end{align*}
for all $m,l\geq N(\epsilon)$. In other words, $\{\varphi_m\}$ is a Cauchy sequence in Banach space $(X_\sigma,\|\cdot\|_\sigma)$, so $\{\varphi_m\}$ converges. Since we construct a convergent subsequence of $\{\psi_m\}$, $P(\Gamma)$ is relatively compact and, hence, $P:\Gamma\rightarrow\Gamma$ is a compact operator.

\begin{center}
\bibliographystyle{alpha}
\bibliography{Bibjournal.bib}

\newcommand{\etalchar}[1]{$^{#1}$}
\begin{thebibliography}{VLVDVK{\etalchar{+}}03}

\bibitem[BC98]{brown1998spatial}
Joel~R Brown and Jennifer Carter.
\newblock Spatial and temporal patterns of exotic shrub invasion in an {A}ustralian tropical grassland.
\newblock {\em Landscape Ecology}, 13:93--102, 1998.

\bibitem[BGP11]{beckage2011grass}
Brian Beckage, Louis~J Gross, and William~J Platt.
\newblock Grass feedbacks on fire stabilize savannas.
\newblock {\em Ecological Modelling}, 222(14):2227--2233, 2011.

\bibitem[CGY17]{chen2017traveling}
Yan-Yu Chen, Jong-Shenq Guo, and Chih-Hong Yao.
\newblock Traveling wave solutions for a continuous and discrete diffusive predator-prey model.
\newblock {\em Journal of Mathematical Analysis and Applications}, 445(1):212--239, 2017.

\bibitem[CH16a]{chen2016maximum}
Chiun-Chuan Chen and Li-Chang Hung.
\newblock A maximum principle for diffusive {L}otka-{V}olterra systems of two competing species.
\newblock {\em Journal of Differential Equations}, 261(8):4573--4592, 2016.

\bibitem[CH16b]{chen2016nonexistence}
Chiun-Chuan Chen and Li-Chang Hung.
\newblock Nonexistence of traveling wave solutions, exact and semi-exact traveling wave solutions for diffusive {L}otka-{V}olterra systems of three competing species.
\newblock {\em Communications on Pure \& Applied Analysis}, 15(4), 2016.

\bibitem[CHH20]{chen2020discrete}
Chiun-Chuan Chen, Ting-Yang Hsiao, and Li-Chang Hung.
\newblock Discrete {N}-barrier maximum principle for a lattice dynamical system arising in competition models.
\newblock {\em Discrete \& Continuous Dynamical Systems: Series A}, 40(1), 2020.

\bibitem[CHL16]{chen2016n}
Chiun-Chuan Chen, Li-Chang Hung, and Chen-Chih Lai.
\newblock An {N}-barrier maximum principle for autonomous systems of n species and its application to problems arising from population dynamics.
\newblock {\em Communications on Pure \& Applied Analysis}, 18, 2016.

\bibitem[CHW]{ChenHsiaoWang2024}
Chiun-Chuan Chen, Ting-Yang Hsiao, and Shun-Chieh Wang.
\newblock The existence of non-monotone traveling wave solution of the coexistence model in preparation.

\bibitem[DLR06]{d2006probabilistic}
Paolo D’Odorico, Francesco Laio, and Luca Ridolfi.
\newblock A probabilistic analysis of fire-induced tree-grass coexistence in savannas.
\newblock {\em The American Naturalist}, 167(3):E79--E87, 2006.

\bibitem[DSM90]{dublin1990elephants}
Holly~T Dublin, Alan~RE Sinclair, and Jacqueline McGlade.
\newblock Elephants and fire as causes of multiple stable states in the {S}erengeti-{M}ara woodlands.
\newblock {\em The Journal of Animal Ecology}, pages 1147--1164, 1990.

\bibitem[FG95]{fonseca1995degree}
Irene Fonseca and Wilfrid Gangbo.
\newblock {\em Degree theory in analysis and applications}.
\newblock Number~2. Oxford University Press, 1995.

\bibitem[GP08]{gordon2008ecology}
Iain~J Gordon and Herbert~HT Prins.
\newblock {\em The ecology of browsing and grazing}.
\newblock Number 195. Springer, 2008.

\bibitem[GP19]{gordon2019ecology}
Iain~J Gordon and Herbert~HT Prins.
\newblock {\em The ecology of browsing and grazing II}.
\newblock Springer, 2019.

\bibitem[GVdVVL17]{groen2017spatially}
Thomas~A Groen, Claudius~ADM Van~de Vijver, and Frank Van~Langevelde.
\newblock Do spatially homogenising and heterogenising processes affect transitions between alternative stable states?
\newblock {\em Ecological modelling}, 365:119--128, 2017.

\bibitem[HLC16]{hung2016n}
Li-Chang Hung, Hsiao-Feng Liu, and Chiun-Chuan Chen.
\newblock N-barrier maximum principle for degenerate elliptic systems and its application.
\newblock {\em Discrete and Continuous Dynamical Systems}, 2016.

\bibitem[Hsi22]{hsiao2022estimates}
Ting-Yang Hsiao.
\newblock Estimates of population size for traveling wave solutions of spatially non-local {L}otka-{V}olterra competition system.
\newblock {\em Journal of Dynamics and Differential Equations}, 34(3):1969--1996, 2022.

\bibitem[KTK23]{klimasara2023model}
Pawe{\l} Klimasara and Marta Tyran-Kami{\'n}ska.
\newblock A model of seasonal savanna dynamics.
\newblock {\em SIAM Journal on Applied Mathematics}, 83(1):122--143, 2023.

\bibitem[Ma01]{ma2001traveling}
Shiwang Ma.
\newblock Traveling wavefronts for delayed reaction-diffusion systems via a fixed point theorem.
\newblock {\em Journal of Differential Equations}, 171(2):294--314, 2001.

\bibitem[Mor75]{morris1975schauder}
Sidney~A Morris.
\newblock The {S}chauder-{T}ychonoff fixed point theorem and applications.
\newblock {\em Matematick{\`y} {\v{c}}asopis}, 25(2):165--172, 1975.

\bibitem[SG84]{sydes1984comparative}
CL~Sydes and JP~Grime.
\newblock A comparative study of root development using a simulated rock crevice.
\newblock {\em The Journal of Ecology}, pages 937--946, 1984.

\bibitem[VA00]{van2000shrub}
Or~W Van~Auken.
\newblock Shrub invasions of {N}orth {A}merican semiarid grasslands.
\newblock {\em Annual Review of Ecology and Systematics}, 31(1):197--215, 2000.

\bibitem[VLVDVK{\etalchar{+}}03]{van2003effects}
Frank Van~Langevelde, Claudius~ADM Van De~Vijver, Lalit Kumar, Johan Van De~Koppel, Nico De~Ridder, Jelte Van~Andel, Andrew~K Skidmore, John~W Hearne, Leo Stroosnijder, William~J Bond, et~al.
\newblock Effects of fire and herbivory on the stability of savanna ecosystems.
\newblock {\em Ecology}, 84(2):337--350, 2003.

\bibitem[WB71]{walter1971ecology}
Heinrich Walter and John~Harrison Burnett.
\newblock {\em Ecology of tropical and subtropical vegetation}, volume 539.
\newblock Oliver and Boyd Edinburgh, 1971.

\bibitem[Zei86]{zeidler1986}
Eberhard Zeidler.
\newblock {\em Nonlinear Functional Analysis and its Applications I: Fixed-Point Theorems}.
\newblock Springer New York, NY, 1986.

\bibitem[Zei12]{zeidler2012applied}
Eberhard Zeidler.
\newblock {\em Applied functional analysis: applications to mathematical physics}, volume 108.
\newblock Springer Science \& Business Media, 2012.

\end{thebibliography}
\end{center}

\end{document}